\documentclass[11pt]{amsart}

\usepackage[a4paper,margin=1in]{geometry} 

\usepackage{amsmath}
\usepackage{mathtools} % Adds cases* and dcases envirnoments
\usepackage[algoruled]{algorithm2e}
\usepackage{enumitem}

\usepackage{tikz}
\usetikzlibrary{arrows,arrows.meta,backgrounds,calc,fit,decorations.pathreplacing,decorations.markings,shapes.geometric}

\tikzstyle{internal} = [draw, fill, shape=circle]
\tikzstyle{external} = [shape=circle]
\tikzstyle{square}   = [draw, fill, rectangle]
\tikzstyle{triangle} = [draw, fill, regular polygon, regular polygon sides=3, inner sep=3pt]
\tikzstyle{pentagon} = [draw, fill, regular polygon, regular polygon sides=5, inner sep=2pt, minimum size=14pt]

\usepackage[margin=1cm]{caption} % Adds an addition margin on either side of figure captions. This must come before "\usepackage{subfig}" (otherwise there are errors).
\usepackage{subfig}

\usepackage[thinlines]{easytable}

\tikzset{every fit/.append style=text badly centered}

\usepackage[bookmarks=true,hypertexnames=false,pagebackref]{hyperref}
\hypersetup{colorlinks=true, citecolor=blue, linkcolor=red, urlcolor=blue}

\usepackage{scrtime}
\usepackage{ifthen}

\usepackage{cleveref}

\usepackage{todonotes}
\usepackage{mleftright}
\usetikzlibrary{positioning,chains,fit,shapes,calc}
\usetikzlibrary{trees}
\usetikzlibrary{decorations.pathreplacing}
\usetikzlibrary{decorations.pathmorphing}
\usetikzlibrary{decorations.markings}
\tikzset{>=latex} % arrow tips

\usepackage{cool}
\Style{DSymb={\mathrm d},DShorten=true,IntegrateDifferentialDSymb=\mathrm{d}}

\newcommand{\numP}{\#{\bf P}}

\newcommand{\vbl}{{\sf var}}

\newcommand{\prs}{partial rejection sampling} % stands out a bit too much?

\newcommand{\Bad}{{\sf Bad}}

\newcommand{\Ex}{\mathop{\mathbb{{}E}}\nolimits}
\newcommand{\Var}{\mathop{\mathrm{Var}}\nolimits}

\def\*#1{\mathbf{#1}}
\def\+#1{\mathcal{#1}}
\def\-#1{\mathrm{#1}}

\newcommand{\abs}[1]{\left\vert#1\right\vert}

\newcommand{\eps}{\varepsilon}

\renewcommand{\Pr}{\mathop{\mathbb{{}P}}\nolimits}

\newcommand\Xbar{\overline X}

\newcommand{\fix}{\mathrm{fix}}

\newcommand{\cycle}{\mathrm{cycle}}
\newcommand{\OmegavblE}{\Omega_E^{\vbl}}
\newcommand{\OmegavblC}{\Omega_{\cycle}^{\vbl}}
\newcommand{\sigmafix}{\sigma_{\fix}}
\newcommand{\maxx}{\mathrm{max}}
\newcommand{\minn}{\mathrm{min}}

\newtheorem{theorem}{Theorem}
\newtheorem{lemma}[theorem]{Lemma}

\newtheorem{proposition}[theorem]{Proposition}
\newtheorem{corollary}[theorem]{Corollary}

\crefname{theorem}{Theorem}{Theorems}
\crefname{observation}{Observation}{Observations}
\crefname{claim}{Claim}{Claims}
\crefname{condition}{Condition}{Conditions}
\crefname{algorithm}{Algorithm}{Algorithms}
\crefname{property}{Property}{Properties}
\crefname{example}{Example}{Examples}
\crefname{fact}{Fact}{Facts}
\crefname{lemma}{Lemma}{Lemmas}
\crefname{corollary}{Corollary}{Corollaries}
\crefname{definition}{Definition}{Definitions}
\crefname{remark}{Remark}{Remarks}
\crefname{proposition}{Proposition}{Propositions}
\crefname{equation}{equation}{equations}

\makeatletter
\def\prob#1#2#3{\goodbreak\begin{list}{}{\labelwidth\z@ \itemindent-\leftmargin
                        \itemsep\z@  \topsep6\p@\@plus6\p@
                        \let\makelabel\descriptionlabel}
                      \item[\textbf{Name}]#1
                      \item[\textbf{Instance}]#2
                      \item[\textbf{Output}]#3
                \end{list}}
\makeatother

\title{Approximately counting bases of bicircular matroids}

\title{Approximately counting bases of bicircular matroids}\thanks{The work described here was supported by the EPSRC research grant
EP/N004221/1 ``Algorithms that Count''.}

\author[H.\ Guo]{Heng Guo}
\address[Heng Guo]{School of Informatics, University of Edinburgh, Informatics Forum, Edinburgh, EH8 9AB, United Kingdom.}
\email{hguo@inf.ed.ac.uk}

\author[M.\ Jerrum]{Mark Jerrum}
\address[Mark Jerrum]{School of Mathematical Sciences,
Queen Mary, University of London, Mile End Road, London, E1 4NS, United Kingdom.}
\email{m.jerrum@qmul.ac.uk}

\begin{document}

\label{firstpage}
\maketitle

\begin{abstract}
  We give a fully polynomial-time randomised approximation scheme (FPRAS) for the number of bases in bicircular matroids.
  This is a natural class of matroids for which counting bases exactly is \numP-hard and yet approximate counting can be done efficiently.
\end{abstract}

\section{Introduction}

We introduce a new application of the ``popping'' paradigm 
that has been used to design efficient perfect samplers for a number of combinatorial structures.
Existing examples are cycle popping~\cite{Wilson96,PW98}, sink popping~\cite{CPP02} and cluster popping~\cite{GP14,GJ19,GH18} that, 
respectively, produce uniformly distributed spanning trees, sink-free orientations in undirected graphs, 
and root-connected subgraphs in directed
graphs (and, as a consequence, connected subgraphs of an undirected graph).  
In doing so we provide an example of a natural class of matroids for which the bases-counting problem is
hard ($\numP$-complete) to solve exactly, but which is polynomial time to solve approximately in the sense
of Fully Polynomial-time Randomised Approximation Schemes (or FPRAS).  For basic definitions connected with the 
complexity of counting problems refer to~\cite{MotRag} or~\cite{ETHmono}.

Towards this end, we introduce ``bicycle popping'' as
a means to sample, uniformly at random, bases of a bicircular matroid.\footnote{``Bicycle popping'' has the
advantage of being easy to remember, but it is important to note that the term is unconnected with the 
concept of bicycle space of a graph.}  
Bicircular matroids are associated with undirected graphs and will be defined in the next section.  
Note that the main result and its proof can be understood in graph-theoretic terms, and
no knowledge of matroid theory is needed beyond the exchange axiom.  Our perfect sampling approach can
be implemented to run in $O(n^2)$ time, where $n$ is the number of vertices 
%and $m$ the number of edges 
in the instance graph. (Refer to Section~\ref{sec:LERW}.) Using a standard reduction, such a sampler can be used to construct an efficient randomised 
algorithm, indeed an FPRAS, for estimating the number of bases within specified relative error (Theorem~\ref{thm:FPRAS} and \Cref{thm:faster-counting}).  

The computational complexity of counting bases of a matroid {\it exactly\/} is still only partially understood.  
According to the class of matroids under consideration, the exact counting problem may be polynomial time,
$\numP$-complete or unresolved.  Counting bases of a graphic matroid (i.e., counting spanning trees of a graph) 
is a classical problem and is well solved by Kirchhoff's matrix-tree theorem.  This method extends 
fairly directly to the wider class of regular matroids~\cite{Maurer76}.  The bases-counting problem for bicircular 
matroids, a restriction of the class of transversal matroids, was shown to be $\numP$-complete by 
Gim\'enez and Noy~\cite{GN06}.  The status of the important case of binary matroids appears to 
be open~\cite{Sno12}. 

Jerrum~\cite{Jer06} showed that it is $\numP$-hard to exactly count bases of certain sparse paving matroids.   
Combined with the approximation algorithm of Ch\'avez Lomel\'\i\ and Welsh \cite{CW96}, %(or by verifying the balanced condition for this particular class of matroids~\cite{Jer06}), 
this result highlights a (presumably) exponential gap between exact and approximate counting. 
However, it could be said that this example is not particularly natural.  Piff and Welsh~\cite{PiffWelsh71}
demonstrated that the number of paving matroids on a ground set of $n$~elements is doubly exponential 
in~$n$, so even representing the problem instance raises significant issues.
%Here we give an efficient popping-style perfect sampler for bases of a bicircular matroid,
%which implies an FPRAS vis standard self-reductions.
Combined with the completeness result of Gim\'enez and Noy \cite{GN06}, our FPRAS provides a more convincing 
and natural demonstration of the gap between exact and approximate counting for matroid bases. 
%On the other hand, our result does not have any implication of the basis-exchange graph of bicircular matroids,
%and it is not clear whether the same goal can be achieved via MCMC.
%In particular, it is not known whether bicircular matroids are balanced.
%Moreover, it is not clear whether the current method generalises to transversal matroids,
%a wider class of matroids containing bicircular matroids.

After posting our paper on arXiv, 
we were made aware of an independent work of Kassel and Kenyon \cite{KK17},
who have proposed essentially the same algorithm\footnote{There are a couple of fine differences, such as not rejecting $2$-cycle, and popping cycles randomly instead of deterministically. These differences do not change the nature of the algorithm.} 
for sampling from a weighted distribution on cycle-rooted spanning trees.
Their interest in the algorithm is as a component in their proofs, 
for which correctness of the algorithm is obviously 
important and is proved in detail.  The time-complexity of their algorithm is 
not analysed in detail, though Kassel and Kenyon offer some brief remarks about the 
run-time of the algorithm on a square grid.  Kassel~\cite{Kassel15} also observes the connection to sampling bases 
of a bicircular matroid, and notes that the corresponding counting problem is \#P-complete.

Even more recently, Anari, Liu, Oveis~Gharan and Vinzant~\cite{ALOV19} 
have shown that the expansion of the so-called ``basis exchange graph'' for any matroid is at least~$1$.  
This result implies that a random walk on the bases exchange graph is rapidly mixing, and provides a Markov Chain Monte Carlo (MCMC) approach to sampling bases of any matroid.  
The mixing time is subsequently sharpened by Cryan, Guo and Mousa~\cite{CGM19}.
The only requirement for the Markov chain approach is that there exists an efficient independence oracle to verify whether a given set is a basis. 
Since this requirement certainly holds for bicircular matroids, these works yield an alternative
approach to sampling bases of a bicircular matroid.  The Markov chain method is very different to ours and does not give a perfect sampler
(though the deviation of the output distribution from uniformity decays exponentially fast in the run-time).
Also, the analysis of the expansion factor of the basis-exchange graph is technically challenging, while the analysis of our popping algorithm is relatively elementary.   
%Although the bases-exchange graph is 
%widely conjectured to be an expander (a special case of the ``zero-one polytope conjecture'' of Mihail and Vazirani \cite{Mihail92,Kaibel04}), 
%there has not been much progress on this conjecture, as far as we are aware.  
Before the work of Anari, Liu, Oveis~Gharan and Vinzant~\cite{ALOV19}, the bases-exchange graph was known to be an expander only in special cases.
Most notably, Feder and Mihail~\cite{FM92} showed that the class of so-called ``balanced matroids'',
a strict superset of the class of regular matroids, has expansion factor at least~$1$.
(See also \cite{JSTV04} for improvements and simplifications.) 
%References: Jerrum, Son, Tetali and Vigoda~\cite{JSTV04}
%and \textbf{connection with Rayleigh matroids/need to do some digging around}.
Furthermore, all paving matroids admit an FPRAS for the number of their bases, 
as shown by Ch\'avez Lomel\'\i\ and Welsh \cite{CW96}, through the straightforward Monte-Carlo method.
%since bases in these matroids are sufficiently dense.

%This result then there is a simple FPRAS that solves the bases counting problem 
%on $M$ via Markov Chain Monte Carlo (MCMC),
%provided that $M$ has an efficient independence oracle to verify whether a given set is a basis. 
%``Reasonable'' here just means that there is an efficient algorithm
%for deciding whether a given subset of the ground set is a basis (so-called ``independence oracle''). 

%The situation with respect to approximation algorithms is even more tantalising.  If the so-called bases-exchange
%graph of a matroid $M$ is a good expander, then there is a simple FPRAS that solves the bases counting problem 
%on $M$ via Markov Chain Monte Carlo (MCMC),
%provided that $M$ has an efficient independence oracle to verify whether a given set is a basis. 
%``Reasonable'' here just means that there is an efficient algorithm
%for deciding whether a given subset of the ground set is a basis (so-called ``independence oracle'').  

\section{Bicycle-popping}

For a graph $G=(V,E)$, let $n=\abs{V}$ and $m=\abs{E}$.
When $m\ge n$ and $G$ is connected, we associate a \emph{bicircular} matroid $B(G)$ with $G$.
The ground set is $E$, and a subset $R\subseteq E$ is independent if every connected component of $(V,R)$ has at most one cycle.
Thus, the set of bases of $B(G)$ is
\begin{align*}
  \+B=\{R\mid\text{every connected component of $(V,R)$ is unicyclic}\}.
\end{align*}
In particular, if $R\in\+B$, then $\abs{R}=n$.
Denote by $\pi_{\+B}(\cdot)$ or simply $\pi(\cdot)$ the uniform distribution over $\+B$.
We refer the reader to \cite{Mat77} for more details on bicircular matroids.
Gim\'enez and Noy~\cite{GN06} have shown that counting the number of bases for bicircular matroids is \numP-complete.
See also \cite{GMN05} for extremal bounds on this number.

We now associate a random arc $a_v=(v,w)$ to each vertex $v\in V$,
which is uniform over all neighbours $w$ of~$v$.
Given an arbitrary assignment $\sigma=(a_v)_{v\in V}$, consider the directed graph $(V,\sigma)$ 
with exactly those $|V|$ arcs.  It is easy to see that 
each (weakly) connected component of this graph has the same number of arcs as vertices.
Thus, there is exactly one (directed) cycle per connected component.  Let $U(\sigma)\subseteq E$
be the subset of edges of~$G$ obtained by dropping the direction of arcs in~$\sigma$.
Consider the distribution $\tau(\cdot)$ on subsets of $E$ induced by $\sigma$ via the mapping $U(\sigma)$.
%We may drop the direction of the arcs, which yields a distribution $\tau(\cdot)$ over subset of edges of size $n$.
There are two reasons why $\tau(\cdot)$ is not quite the same as $\pi(\cdot)$.
\begin{enumerate}
  \item It is possible to have $2$-cycles in $\sigma$, in which case at least one connected component of $U(\sigma)$ will be a tree rather than a 
  unicyclic graph.
  \item Every cycle in $\sigma$ of length greater than 2 may be reversed without changing $U(\sigma)$.  
  Thus, in $\tau(\cdot)$, each subgraph with $k$ connected components arises in $2^k$ ways, skewing the distribution towards 
  configurations with more connected components.  
\end{enumerate}
For each edge $e\in E$, let $B_e$ denote the event that a $2$-cycle is present at $e$, i.e., both orientations of $e$ appear in~$\sigma$.
For each cycle~$C$ in~$G$, we fix an arbitrary orientation and denote by $B_C$ the event that $C$ is oriented this way in~$\sigma$.
If we further condition on none of $B_e$ or $B_C$ happening, the resulting distribution~$\tau$ induced by $U(\sigma)$ is exactly $\pi(\cdot)$.

Partial rejection sampling \cite{GJL19} provides a useful framework 
to sample from a product distribution conditioned on a number of bad events not happening.
In particular, we call a collection of bad events \emph{extremal} 
if any two bad events are either probabilistically independent or disjoint (i.e., cannot both occur).
It is straightforward to verify that the collection of bad events $\{B_e\mid e\in E\}\cup\{B_C\mid \text{$C$ is a cycle in $G$}\}$ is extremal. 
(The reason is similar to the cycle-popping algorithm. See \cite[Section 4.2]{GJL19}.
In fact, the bad events here are either identical or more restrictive than those for cycle-popping.)
For an extremal instance, to draw from the desired distribution,
we only need to randomly initialize all variables,
and then repeatedly re-randomize variables responsible for occurring bad events.
This is \Cref{alg:clock-popping}, which we call ``bicycle-popping''.

\begin{algorithm}
  \caption{Bicycle-popping}
  \label{alg:clock-popping}
  Let $S$ be a subset of arcs obtained by assigning each arc $a_v$ independently and uniformly among all neighbours of $v$\;
  \While{a bad event $B_e$ or $B_C$ is present}{
    Let $\Bad$ be the set of vertices that are contained in any edge $e$ or cycle $C$ such that $B_e$ or $B_C$ occurs\;
    Re-randomize $\{a_v\mid v\in \Bad\}$ to get a new $S$\;
  }
  \KwRet{the undirected version of $S$}
\end{algorithm}

We need to be a little bit careful about bad events $(B_C)$,
since there are potentially exponentially many cycles in $G$.
We cannot afford to dictate the unfavourable orientation a priori,
but rather need to figure it out as the algorithm executes.
This is not difficult to get around, since we only need an arbitrary (but deterministic) orientation of each cycle.
For example, we may arbitrarily order all vertices, and give a sign $\pm$ to each direction of an edge according to the ordering.
The sign of an odd-length cycle is the product over all its edges,
and the sign of an even-length cycle is the product over all but the least indexed edge.
Then, we can simply declare all orientations with a $+$ sign ``bad''.
An alternative is to reject cycles randomly, which is considered in \cite{KK17} and is described in \Cref{sec:Gibbs}.

Since the extremal condition is satisfied, applying \cite[Theorem 8]{GJL19} we get the correctness of \Cref{alg:clock-popping}.

\begin{proposition}\label{prop:correctness}
  Conditioned on terminating, the output of \Cref{alg:clock-popping} is exactly $\pi(\cdot)$.
\end{proposition}

We remark that bicircular popping, \Cref{alg:clock-popping}, 
differs from cycle-popping \cite{PW98} by associating random variables to \emph{all} vertices,
and differs from cluster-popping \cite{GP14,GJ19} by associating random variables to vertices rather than edges.

\section{Run-time analysis}
\label{sec:run-time}

An advantage of adopting the \prs\ framework is that 
we have a closed form formula for the expected run-time of these algorithms on extremal instances.
%See \cite[Lemma 12, Theorem 13]{GJL17} and \cite[Eqn.~(2)]{GH18}.

In the general setting of partial rejection sampling,
the target distribution to be sampled from is a product distribution over variables, 
conditioned on a set of ``bad'' events $(A_i)_{i\in \+I}$ not happening for some index set $\+I$.
Let $T_i$ be the number of resamplings of event $A_i$.
Let $q_i$ be the probability such that exactly $A_i$ occurs,
and $q_{\emptyset}$ be the probability such that none of $(A_i)_{i\in \+I}$ occurs,
both under the product distribution.
Suppose $q_{\emptyset}>0$ as otherwise the support of $\pi(\cdot)$ is empty.
For extremal instances, \cite[Lemma 12]{GJL19} and the first part of the proof of \cite[Theorem 13]{GJL19} yield
\begin{align}\label{eqn:expected-i}
  \Ex T_i=\frac{q_i}{q_{\emptyset}}.
\end{align}
Let $T$ be the number of resampled variables.
By linearity of expectation and \eqref{eqn:expected-i},
\begin{align}\label{eqn:expected}
  \Ex T=\sum_{i\in \+I} \frac{q_i\cdot\abs{\vbl(A_i)}}{q_{\emptyset}}.
\end{align}
(See also \cite[Eqn.~(2)]{GH18}.)
We note that an upper bound similar to the right hand side of \eqref{eqn:expected} was first shown by Kolipaka and Szegedy \cite{KS11},
in a much more general setting but counting only the number of resampled events.

Specialising to \Cref{alg:clock-popping}, 
let $q_e$ and $q_C$ be the corresponding quantity for bad events $B_e$ and $B_C$, respectively.
Let $\Omega_0$ be the set of assignments so that no bad event happens,
and $\Omega_e$ (or $\Omega_C$) be the set of assignments of $(a_v)_{v\in V}$ so that exactly $B_e$ (or $B_C$) happens
and none of the other bad events happen.
Then $\abs{\Omega_0}=\abs{\+B}$.
For a bad event $B$, let $\vbl(B)$ be the set of variables defining $B$,
namely, $\vbl(B_e)=\{a_u,a_v\}$ if $e=(u,v)\in E$,
and $\vbl(B_C)=\{a_v\mid v\in C\}$ if $C$ is a cycle in $G$.
Define 
\begin{align*}
  \OmegavblE:=\{\left( \sigma,a_v \right)\mid \exists e\in E,\;\sigma\in\Omega_e,\;a_v\in\vbl(B_e)\},
\end{align*}
and 
\begin{align*}
  \OmegavblC:=\{\left( \sigma,a_v \right)\mid \exists\text{a cycle $C$},\;\sigma\in \Omega_C,\;a_v\in\vbl(B_C)\}.
\end{align*}
Then $\sum_{e\in E} \frac{q_e\cdot\abs{\vbl(A_i)}}{q_{\emptyset}}=\frac{\big|\OmegavblE\big|}{\abs{\Omega_0}}$
and $\sum_{C\text{ is a cycle}} \frac{q_C\cdot\abs{\vbl(A_i)}}{q_{\emptyset}}=\frac{\big|\OmegavblC\big|}{\abs{\Omega_0}}$. 

\begin{proposition}  \label{prop:extremal:runningtime}
  Let $T$ be the number of resampled variables of \Cref{alg:clock-popping}.
  Then 
  \begin{align*}
    \Ex T = \frac{\big|\OmegavblE\big|}{\abs{\Omega_0}} + \frac{\big|\OmegavblC\big|}{\abs{\Omega_0}}.
  \end{align*}
\end{proposition}

%A related upper bound of \Cref{prop:extremal:runningtime} was first shown by Kolipaka and Szegedy~\cite{KS11} in the general Lov\'asz local lemma setting.

We bound these ratios using a combinatorial encoding idea.
Namely, we want to design an injective mapping from $\OmegavblE$ or $\OmegavblC$ to $\Omega_0$.
To make the mapping injective, we in fact have to record some extra information.
We first deal with $\OmegavblC$.

\begin{lemma}\label{lem:magic-ratio-cycle}
  For a connected graph $G=(V,E)$ with $m\ge n$ where $m=\abs{E}$ and $n=\abs{V}$,
  $\big|\OmegavblC\big| \le n \abs{\Omega_0}$.
\end{lemma}
\begin{proof}
  We define a ``repairing'' mapping $\varphi:\OmegavblC\rightarrow\Omega_0\times V$, as follows.
  For $\sigma\in\Omega_C$,
  we define $\sigmafix$ to be the same as $\sigma$ except that the orientation of $C$ is reversed.
  Clearly $\sigmafix\in\Omega_0$.
  Let 
  \begin{align*}
    \varphi(\sigma,a_v) = (\sigmafix,v) \quad \text{if $\sigma\in\Omega_C$ and $v\in C$}.
  \end{align*}
  We claim that $\varphi$ is injective.
  To see this, given $\sigmafix$ and $v$, we simply flip the orientations of the cycle containing $v$ to recover $\sigma$.
  Since $\varphi$ is injective, we have that $\big|\OmegavblC\big| \le n \abs{\Omega_0}$.
\end{proof}

For $\OmegavblE$, the proof is slightly more involved.
For $\sigma\in\Omega_e$, if we contract $e$, this component is a directed tree rooted at $e$,
where all edges are directed toward $e$.

\begin{lemma}\label{lem:magic-ratio-edge}
  For a connected graph $G=(V,E)$ with $m\ge n$ where $m=\abs{E}$ and $n=\abs{V}$,
  $\big|\OmegavblE\big| \le 2n(n-1) \abs{\Omega_0}$.
\end{lemma}
\begin{proof}
  Let $\Omega_E:=\bigcup_{e\in E}\Omega_e$.
  Then $\big|\OmegavblE\big| = 2\abs{\Omega_E}$.

  Fix an arbitrary ordering of all vertices and edges.
  Our goal to define an injective ``repairing'' mapping $\varphi:\Omega_E\rightarrow\Omega_0\times V\times E$.
  For $\sigma\in\Omega_e$,
  find the connected component of $U(\sigma)$ containing the edge $e=(v_1,v_2)$, and let its vertex set be~$S$. 
  Depending on whether $S=V$, there are two cases.
  \begin{enumerate}
    \item If $S\neq V$, then, since the graph~$G$ is connected,
      there must be at least one edge joining the component to the rest of the graph.
      Pick the first such edge $(u,u')$ where $u$ is in $S$ and $u'$ is not.
    \item Otherwise $S=V$; then since the graph has at least $n$ edges,
      there must be at least one edge not in $U(\sigma)$.
      Let $e'=(u,u')$ be the first such edge, and $C$ be the cycle resulting from adding $e'$ to $U(\sigma)$.
      Suppose the correct orientation on $C$ induces the orientation $u\rightarrow u'$ on~$e'$. %where $(u,u')\in E$.
      % Didn't understand (u,u') \in E?
  \end{enumerate}
  Let $u=u_1,u_2,\dots,u_{\ell}=v_1$ be the unique path between $u$ and $v_1$ in $U(\sigma)$.
  (The vertex $v_1$ is chosen arbitrarily from the two endpoints of $e$.)
  Let $\sigmafix$ be the assignment so that $a_{u_i}$ points to $u_{i-1}$, where $u_0=u'$, and all other 
  variables are unchanged from~$\sigma$.
%  Since $B_e$ is the only bad event occurring under $\sigma$,
  It is easy to verify that $\sigmafix\in\Omega_0$.
  Also, $\sigmafix$ does not depend on the choice of $v_1$ from the edge $e$.
  Define $\varphi(\sigma)=(\sigmafix,u,e)$ where $e=(v_1,v_2)$.

  We claim that $\varphi$ is injective.
  We just need to recover $\sigma$ given $(\sigmafix,u,e)$.
  We first figure out whether $S=V$.
  Notice that $u'$ can be recovered as $u\rightarrow u'$ is in $\sigmafix$.
  If $S\neq V$, then the edge $(u,u')$ is a bridge under $\sigmafix$;
  whereas if $S=V$, $(u,u')$ is not.

  In the first case, simply find the path between $u$ and $v_1$,
  and reverse the ``repairing'' to yield the original $\sigma$.
  In the second case, we remove $(u,u')$ first,
  and then recover the unique path between $u$ and $v_1$.
  The rest is the same as the first case.

  Note that $\abs{\sigmafix}=n$, $u\rightarrow u'\in\sigmafix$, and $e\in U(\sigmafix)$,
  but $(v_1,v_2)\neq(u,u')$.
  Thus, fixing $\sigmafix$, there are $n$ choices for $u$, and $(n-1)$ choices for $e=(v_1,v_2)$.
  Since $\varphi$ is injective, we have that $\big|\OmegavblE\big| = 2\abs{\Omega_E} \le 2n(n-1) \abs{\Omega_0}$.
\end{proof}

Combining Lemma~\ref{lem:magic-ratio-cycle}, Lemma~\ref{lem:magic-ratio-edge}, and Proposition~\ref{prop:extremal:runningtime},
we have the following theorem.

\begin{theorem}  \label{thm:clock-popping}
  Let $G=(V,E)$ be a connected graph, $n=\abs{V}$, $m=\abs{E}$, and $m\ge n$.
  The expected number of random variables sampled in \Cref{alg:clock-popping} on $G$ is at most $2n^2-n$.
\end{theorem}

The bound in \Cref{thm:clock-popping} is tight.
Consider a cycle of length $n$.
Clearly $\abs{\Omega_0}=1$ and $\big|\OmegavblC\big|=n$ as there is only one cycle containing $n$ edges.
Moreover, $\abs{\Omega_e}=n-1$ for there are $n-1$ choices of the missing edge.
Thus $\big|\OmegavblE\big| = 2\abs{\Omega_E} = 2n(n-1)$ and the upper bound is achieved.

\section{An implementation based on a loop-erasing random walk}
\label{sec:LERW}

In the execution of \Cref{alg:clock-popping},
during each iteration, one needs to find all bad events,
and a naive implementation may take up to $O(n)$ time for this task,
giving another factor on top of the bound in \Cref{thm:clock-popping}.
Here we provide an implementation that has expected run-time $O(n^2)$, similar to the loop-erasing random walk of Wilson \cite{Wilson96}.
A formal description is given in \Cref{alg:BLERW}.

\begin{algorithm}
  \caption{A random walk implementation of bicycle-popping}
  \label{alg:BLERW}
  $V_u\gets V$\;
  $S\gets \emptyset$\;
  \While{$V_u\neq\emptyset$}{
    $v\gets$ an arbitrary vertex in $V_u$\;
    Start a random walk from $v$, where in each step we move uniformly at random to a neighbour of the current vertex.
    Erase any cycle $C$ having length $2$ or a wrong orientation,
    until some vertex in $V\setminus V_u$ is reached, or a good cycle $C$ is formed\;
    Remove all vertices of the walk from $V_u$\;
    Add all (undirected) edges along the walk to $S$\;
  }
  \KwRet{$S$}
\end{algorithm}

Observe that, in \Cref{alg:BLERW}, once a cycle is orientated correctly,
none of its associated arcs will be resampled ever again,
and the same holds for any arc attached to it.
We will call such arcs ``fixed''.
Starting from an arbitrary vertex~$v$, we assign a random arc from $v$ to~$u$,
and continue this for~$u$.
So far this is just the normal random walk with memory.
The difference is that whenever a cycle appears, we check whether it has length $>2$ and the correct orientation.
If not, then we erase it, and continue the random walk.
Otherwise, we keep all random arcs leading towards this cycle, and mark them as fixed.
Thus, \Cref{alg:BLERW} amounts to a loop-erasing random walk with a special erasing rule.

Once the first random walk stops with a correctly oriented cycle,
we do the same for the next vertex that has not been fixed yet.
Now the new walk has two possible terminating conditions.
Namely it is fixed if it has reached some fixed vertex, or a correctly oriented cycle of length $>2$ is formed.
This process is repeated until all vertices are fixed.

%One key observation is that since the extremal condition is met, 
%the order of resampled events actually does not matter \cite{GJL17}.
\cref{alg:clock-popping} specifies a particular
order of resampling bad events, modulo the ordering of bad events within each iteration of the while-loop. 
However, bad events can be sampled in any order, without affecting correctness or the expected number of resampled variables. 
%One way to see this is through Newman's lemma \cite{Newman42} by viewing resampling as an ``abstract'' rewriting system.
Although the proof of this key fact has appeared in the context of specific instances of \prs, 
such as cycle-popping \cite{PW98} and sink-popping \cite{CPP02}, we are not aware that the argument has been 
presented in generality, so we do so presently.
As a consequence of this key fact, \Cref{alg:BLERW}, which is sequential, has the same resulting distribution and expected number of resampled variables as \Cref{alg:clock-popping}, which is parallel.
In particular, the expected run-time of \Cref{alg:BLERW} has the same order as the number of resampled variables,
which is at most $O(n^2)$ by \Cref{thm:clock-popping}.

The correctness of \Cref{alg:BLERW} is due to the aforementioned fact that the ordering of resamplings does not matter for extremal instances.  We now formalise and verify this fact.
Consider a generic \prs\ algorithm that repeatedly locates an occurring bad event and resamples the variables on which it depends.
A specific implementation will choose a particular order for resampling the bad events.  We can represent the choices made as a path in a countably infinite, directed ``game graph'' $\Gamma=(\Sigma,A)$.  The vertex set $\Sigma$ of $\Gamma$ contains all multisets of bad events.  We refer to these vertices as \emph{states}.
The arc set~$A$ is defined relative to a \emph{resampling table}, as used in \cite{GJL19}, following Moser and Tardos.
As the algorithm proceeds, the ``frontier'' in the table between used and fresh random variables advances;  in the notation of \cite{GJL19}, the frontier at time~$t$ is specified by the indices $(j_{i,t}:1\leq i\leq n)$.  At time~$t$, the implementation will have sampled a certain multiset $M\in\Sigma$ of bad events:  an event $B_*$ that has been resampled $k$~times will occur $k$ times in~$M$.  Note that $M$ determines the number of times each variable has been resampled, and hence the frontier of the table.  So, even though we don't know the order in which those bad events were resampled, we do know the occurring bad events at time~$t$.  For each $M\in\Sigma$ and each possible occurring bad event~$B_*$, we add an arc in $\Gamma$ from $M$ to $M'=M+B_*$.  A state with outdegree~0 is a \emph{terminating state}.  Given a fixed resampling table, an implementation of \prs\ will generate a directed path in $\Gamma$ starting at the state~$\emptyset$.  With probability~1 (over the choice of resampling table), this path will be finite, i.e., end in a terminating state.  

We now apply a Lemma of Eriksson~\cite{Eriksson96}, which is similar in spirit to Newman's Lemma, but which is both more elementary and better suited to our needs.  Observe that if two bad events occur at time~$t$ then they can be resampled in either order without altering the result;  this is a consequence of the fact that the events are on disjoint sets of variables.  In the terminology of~\cite{Eriksson96}, the game graph $\Gamma$ has the \emph{polygon property}.  
It follows from his Theorem 2.1 that $\Gamma$ has the \emph{strong convergence property}:  
if there exists a path starting at $\emptyset$ and terminating at $M$, then every path starting at $\emptyset$ will terminate at~$M$ in the same number of steps.  Since a terminating path exists with probability~1, we see that both the output and the number of resampled variables is independent of the order in which the implementation decides to resample bad events.  In other words, the correctness of \Cref{alg:BLERW} follows from that of \Cref{alg:clock-popping}, and the distribution of the number of resampled variables is identical in the two algorithms.

\section{Approximating the number of bases}

For completeness, we include a standard self-reduction to count the number of bases of a bicircular matroid,
utilising \Cref{alg:clock-popping}.

\begin{theorem}\label{thm:FPRAS}
  There is an FPRAS for counting bases of a bicircular matroid, with time complexity $O(n^3m^2\eps^{-2})$.
\end{theorem}

\begin{proof}
Let $0<\eps<1$ be a parameter expressing the desired accuracy.
Also let $N(G)$ be the number of bases of $B(G)$, the bicircular matroid associated with $G$.

The technique for reducing approximate counting to sampling is entirely standard~\cite[Chap.~3]{ETHmono}, but we include the 
details here for completeness.
Fix any sequence of graphs $G=G_m,G_{m-1},\ldots,\allowbreak G_{n+1},G_n$, where each graph $G_{i-1}$ is obtained from the previous one $G_i$ by removing a single edge~$e_i$,
and $G_n$ is a disjoint union of unicyclic components.  (Thus the edge set of $G_n$ is a basis of $B(G)$.) Then, noting 
$N(G_n)=1$,
\begin{equation}\label{eq:telescope}
N(G)^{-1}=N(G_m)^{-1}=\frac{N(G_{m-1})}{N(G_m)}\times \frac{N(G_{m-2})}{N(G_{m-1})}\times\cdots\times\frac{N(G_{n+1})}{N(G_{n+2})}\times\frac{N(G_n)}{N(G_{n+1})}.
\end{equation}
Let $X_i$ be the random variable resulting from the following trial: select, uniformly at random, a basis $R$ from $B(G_i)$ and set
$$
X_i=\begin{cases}1,&\text{if $e_i\notin R$;}\\0,&\text{otherwise}.\end{cases}
$$
Here we use \Cref{alg:BLERW} to generate the uniform random basis $R$.
For different $i$, we use fresh random sources so that all $X_i$ are mutually independent.
Note that $\mu_i=\Ex X_i=N(G_{i-1})/N(G_i)$, so that 
$$N(G)^{-1}=\Ex(X_mX_{m-1}\ldots X_{n+2}X_{n+1})=\mu_m\mu_{m-1}\ldots\mu_{n+2}\mu_{n+1}.$$ 
%where the $m-n$ trials are assumed independent.  

Now let $\Xbar_i$ be obtained by taking the mean of $t$ independent copies of the random variable~$X_i$.  
Since $\Ex\Xbar_i=\mu_i$, we have $N(G)^{-1}=\Ex Z$, where $Z=\Xbar_m\Xbar_{m-2}\cdots\Xbar_{n+1}$.  
Also, $\Var\Xbar_i=t^{-1}\Var X_i$, so if $t$ is large enough the variance
of~$Z$ will be small, and $Z^{-1}$ will be a good estimate for $N(G)$.
For this approach to yield a polynomial-time algorithm, we need that all the fractions appearing in the product~\eqref{eq:telescope} are not too small.  
In fact we will show that they are all bounded below by $1/2n$, which is sufficient. 

For the moment, assume this claim, i.e., that $1/2n\leq\mu_i\leq1$, for all $n<i\leq m$.
Note that $\Ex X_i^2=\Ex X_i=\mu_i$ since $X_i$ is a 0,1-variable.  
Standard manipulations give 
$$\Ex \Xbar_i^2=\Var\Xbar_i+(\Ex\Xbar_i)^2=t^{-1}\Var X_i+\mu_i^2=t^{-1}(\Ex X_i^2-\mu_i^2)+\mu_i^2\leq\mu_i^2\left(1+\frac{1}{t\mu_i}\right),$$
whence
$$
\Ex Z^2=\Ex\Xbar_m^2\ldots\Ex\Xbar_{n+1}^2 \leq \mu_m^2\ldots\mu_{n+1}^2\Big(1+\frac{2n}t\Big)^m=(\Ex Z)^2\Big(1+\frac{2n}t\Big)^m
$$
Setting $t=40nm\eps^{-2}$, we obtain $\Ex Z^2=\exp(\eps^2/20)(\Ex Z)^2\leq (1+\eps^2/16)(\Ex Z)^2$, which implies $\Var Z\leq(\eps^2/16)(\Ex Z)^2$.
Thus, by Chebyshev's inequality,
$$
\Pr\big[|Z-N(G)^{-1}|\leq\tfrac12\eps N(G)^{-1}\big]=\Pr\big[|Z-\Ex Z|\leq\tfrac12\eps\Ex Z\big]\geq\tfrac34.$$
It follows that 
$\Pr\big[|Z^{-1}-N(G)|\leq\eps N(G)\big]\geq\frac34$.  In other words, the algorithm that returns 
the estimate $Z^{-1}$ satisfies the conditions for an FPRAS.

To complete the proof we just need to bound the ratio $\mu_i=N(G_{i-1})/N(G_i)$.  
Let $R$ be a basis of $B(G_i)$ that contains the edge $e_i$, i.e., that is not a basis of $B(G_{i-1})$.  
Let $R_0$ be the unique basis in $B(G_n)$ and note that $e_i\notin R_0$.
Since $R_0$ is also a basis of $B(G_i)$, the exchange axiom for matroids asserts that there is an edge $f\in R_0\setminus R$ such that $R+f-e_i$ is a basis of $B(G_i)$ and hence of $B(G_{i-1})$.  
This exchange operation associates a basis in $B(G_{i-1})$ with each basis in $B(G_i)$ that is not a basis in $B(G_{i-1})$; furthermore,  every basis in $B(G_{i-1})$ arises at most $|R_0|=n$ times in this way.  It follows that $N(G_i)\leq (n+1)N(G_{i-1})\leq 2nN(G_{i-1})$, as required.

Overall we need $O(nm\eps^{-2})$ samples for $m-n$ estimators each.
For each sample, we use \Cref{alg:BLERW} which has expected run-time $O(n^2)$ by \Cref{thm:clock-popping}, yielding the claimed time complexity.
\end{proof}

\section{Faster approximate counting}
\label{sec:Gibbs}

Similar to \cite{KK17}, let $\Omega$ be the set of configurations consisting of directed edges so that every vertex is the tail of exactly one arc. % tail, according to Wikipedia?
Consider the following Gibbs distribution:
\begin{align}\label{eqn:Gibbs}
  \rho_{\gamma_2,\gamma}(S)\propto\gamma_2^{C_2(S)}\gamma^{C(S)},
\end{align}
where $S\in\Omega$, $\gamma_2,\gamma\ge 0$ are two parameters,
and $C(S)$ (or $C_2(S)$) is the number of cycles of length greater than~2 (or $2$-cycles) present in $(V,S)$.
We adopt the convention that $0^0=1$.
Then \Cref{alg:clock-popping} and \Cref{alg:BLERW} sample from the distribution $\mu_{0,0.5}$.
We define also the corresponding partition function
\begin{align}\label{eqn:Gibbs-PF}
  Z_{\gamma_2,\gamma}(G)=\sum_{S\in\Omega}\gamma_2^{C_2(S)}\gamma^{C(S)}.
\end{align}
Then $Z_{0,0.5}(G)=\abs{\+B}$.
Another interesting special case is $\rho_{1,1}$,
with the corresponding partition function $Z_{1,1}(G)=\prod_{v\in V}\deg(v)$.
This is because $\rho_{1,1}$ corresponds to choosing a random neighbour of $v$ for each $v\in V$.
Note that each cycle longer than $2$ has two potential orientations.

In order to sample from $\rho_{\gamma_2,\gamma}$ in \eqref{eqn:Gibbs} where $\gamma_2,\gamma\in[0,1]$,
we introduce the following variant of \Cref{alg:BLERW}.
Again, this is very similar to the sampling algorithm of Kassel and Kenyon \cite{KK17}.

\begin{algorithm}
  \caption{Sample $\rho_{\gamma_2,\gamma}$ in \eqref{eqn:Gibbs}}
  \label{alg:Gibbs-sampling}
  $V_u\gets V$\;
  $S\gets \emptyset$\;
  \While{$V_u\neq\emptyset$}{
    $v\gets$ an arbitrary vertex in $V_u$\;
    Start a random walk from $v$, and erase any cycle $C$ formed with probability $1-\gamma_2$ if the length is $2$ or with probability $1-\gamma$ otherwise,
    until some vertex in $V\setminus V_u$ is reached, or a cycle $C$ is accepted\;
    Remove all vertices of the walk from $V_u$\;
    Add all arcs along the walk to $S$\;
  }
  \KwRet{$S$}
\end{algorithm}

The correctness of \Cref{alg:Gibbs-sampling} also follows from \cite[Theorem 8]{GJL19} and the argument in \Cref{sec:LERW}.
We introduce an auxiliary variable for each cycle, which is false with probability $1-\gamma_2$ if the cycle has length $2$, 
or with probability $1-\gamma$ if the cycle is longer.
A cycle is ``bad'' if and only if it is present and the auxiliary variable is false.
Although there are exponentially many such auxiliary variables, 
we only reveal them when necessary.
Every time a cycle is popped, the auxiliary variable is reset.
By the same reason as for \Cref{alg:clock-popping}, such an instance is \emph{extremal} and the correctness follows.

Since the extremal condition holds, the run-time of \Cref{alg:Gibbs-sampling} can be analysed analogously to that of \Cref{alg:clock-popping} and \Cref{alg:BLERW}.
Let $T$ be the number of resampled variables of \Cref{alg:Gibbs-sampling}.
We apply \eqref{eqn:expected}.
First consider the case of $\gamma_2=0$. 
To bound $\sum_{e\in E}\frac{q_e\abs{e}}{q_{\emptyset}}$, we use the injective mapping in Lemma~\ref{lem:magic-ratio-edge},
and to bound $\sum_{C\text{ is a cycle}}\frac{q_C\abs{C}}{q_{\emptyset}}$,
we use an injective mapping similar to the one in Lemma~\ref{lem:magic-ratio-cycle} by simply flipping the auxiliary variable.
Observe that both mappings preserve all cycles other than the one repaired,
and as a consequence, preserve all weights up to a factor $\frac{1-\gamma}{\gamma}$ in the latter case.
It implies that when $\gamma_2=0$, the run-time can be bounded as follows:
\begin{align}\label{eqn:gamma_2=0}
  \Ex T \le 2n(n-1) + \frac{1-\gamma}{\gamma}\cdot n.
\end{align}

Otherwise $\gamma_2>0$.
To bound $\sum_{e\in E}\frac{q_e\abs{e}}{q_{\emptyset}}$,
again, we need to resort to the injective mapping in Lemma~\ref{lem:magic-ratio-edge}.
However, now the ``perfect'' configurations and ``one-flaw'' configurations allow more than one $2$-cycles.
Let $\Omega^{(k)}\subset\Omega$ be the set of configurations with $k$ $2$-cycles, for any $k\le n/2$.

\begin{lemma}
  There is an injective mapping $\psi_k:\Omega^{(k)}\rightarrow\Psi^{(k-1)}$ for any $k\ge 1$,
  where $\Psi^{(k-1)}:=\{(\sigma,v,e)\mid \sigma\in\Omega^{(k-1)},~v\in V,~e\in U(\sigma)\}$.
  Moreover, $\psi_k$ preserves all cycles except one with length $2$.
  \label{lem:2-cycles-ratio}
\end{lemma}
\begin{proof}
  The mapping is similar to the one in Lemma~\ref{lem:magic-ratio-edge}.
  Given $S\in\Omega^{(k)}$,
  we choose an arbitrary 2-cycle,
  and apply the ``fix'' of the injective mapping in Lemma~\ref{lem:magic-ratio-edge}.
  It is straightforward to check that this operation will only destroy the chosen $2$-cycle, 
  and it is reversible given the auxiliary information.
\end{proof}

A consequence of Lemma~\ref{lem:2-cycles-ratio} is that $\sum_{e\in E}\frac{q_e\abs{e}}{q_{\emptyset}}\le 2n^2$,
similar to Lemma~\ref{lem:magic-ratio-edge}.
This is because for any configuration with exactly one ``bad'' $2$-cycle,
applying the mapping in Lemma~\ref{lem:2-cycles-ratio} yields a configuration without the presence of the bad $2$-cycle,
and other cycle structures are all preserved.
The overhead is to remember one vertex and one edge from the undirected version of the image.

To bound $\sum_{C\text{ is a cycle}}\frac{q_C\abs{C}}{q_{\emptyset}}$,
consider again the injective mapping similar to the one in Lemma~\ref{lem:magic-ratio-cycle},
by simply flipping the value of the auxiliary variable.
It implies that $\sum_{C\text{ is a cycle}}\frac{q_C\abs{C}}{q_{\emptyset}}\le \frac{1-\gamma}{\gamma}\cdot n$.
Thus we have that for $\gamma_2>0$, the following also holds:
\begin{align}\label{eqn:gamma_2>0}
  \Ex T \le 2n^2 + \frac{1-\gamma}{\gamma}\cdot n.
\end{align}

Combining the two cases \eqref{eqn:gamma_2=0} and \eqref{eqn:gamma_2>0}, we have the following corollary.

\begin{corollary}\label{cor:Gibbs-RT}
  The run-time of \Cref{alg:Gibbs-sampling}, in expectation, is $O(n^2)$ if $\gamma_2\in[0,1]$ and $\gamma\in[1/2,1]$.
\end{corollary}

The Gibbs formulation \eqref{eqn:Gibbs} allows us to utilise faster annealing algorithms to reduce approximate counting to sampling.
See \cite{SVV09, Hub15}, and the current best algorithm is due to Kolmogorov \cite{Kol18}.

In fact, we will need a slight generalisation from \cite{GH18} as follows.
Let $\Omega$ be a finite set, and the generalised Gibbs distribution $\rho_{\beta}(\cdot)$ over $\Omega$ takes the following form:
\begin{align}\label{eqn:rho}
  \rho_{\beta}(X) = \frac{1}{Z(\beta)}\exp(-\beta H(X))\cdot F(X),
\end{align}
where $\beta$ is the \emph{temperature}, $H(X)\ge0$ is an integer function called the \emph{Hamiltonian},
$F:\Omega\rightarrow\mathbb{R}^+$ is a non-negative function,
and, with a little abuse of notation, $Z(\beta)=\sum_{X\in\Omega}\exp(-\beta H(X))\cdot F(X)$ is the normalising factor.
We would like to turn the sampling algorithm into an approximation algorithm to $Z(\beta)$.
Typically, this involves calling the sampling oracle in a range of temperatures, which we denote $[\beta_{\minn},\beta_{\maxx}]$.
(This process is usually called simulated annealing.)
Let $Q:=\frac{Z(\beta_{\minn})}{Z(\beta_{\maxx})}$, $q=\log Q$, and $N=\max_{X\in\Omega} H(X)$.
The following result is due to Kolmogorov \cite[Theorem 8]{Kol18},
as extended in \cite[Lemma 8]{GH18}.

\begin{proposition}  \label{prop:Kolmogorov}
  Suppose we have a sampling oracle from the distribution $\rho_{\beta}$ for any $\beta\in[\beta_{\minn},\beta_{\maxx}]$.
  There is an algorithm to approximate $Q$ within $1\pm\eps$ multiplicative errors
  using $O(q\log N/\eps^{2})$ oracle calls in average.
\end{proposition}

\begin{theorem}\label{thm:faster-counting}
  There is an FPRAS for counting bases of a bicircular matroid, with time complexity $O(n^3(\log n)^2\eps^{-2})$.
\end{theorem}
\begin{proof}
  We will do a two stage annealing.
  Our starting point is $Z_{1,1}(G)=\prod_{v\in V}\deg(v)$.
  We first apply the annealing algorithm in Proposition~\ref{prop:Kolmogorov} between $\gamma_2=1$ and $\gamma_2=0$.
  We only treat $\gamma_2$ as the temperature in this stage but not $\gamma$.\footnote{This explains why we need $F(X)$ in \eqref{eqn:rho}.}
  Clearly $H(S)=C_2(S)\le m$ in this case, and $\log N=O(\log n)$.
  The more complicated estimate is $Q=\frac{Z_{1,1}(G)}{Z_{0,1}(G)}$.
  For any $S\in\Omega$, we apply the mapping in Lemma~\ref{lem:2-cycles-ratio} at most $C_2(S)\le n$ times to get $S'$ where no $2$-cycle is present.
  Given $S'$, and a sequence of vertex-edge pair generated by this repairing sequence,
  we may uniquely recover $S$ since the mapping in Lemma~\ref{lem:2-cycles-ratio} is injective.
  Moreover, the longer cycles in $S$ are all preserved in $S'$.
  It implies that $Q=\frac{Z_{1,1}(G)}{Z_{0,1}(G)}\le (mn)^n$ and $q=\log Q=O(n\log n)$.
  Hence, the number of samples required in this step is $O(n(\log n)^2\eps^{-2})$.  

  In the second stage, we apply Proposition~\ref{prop:Kolmogorov} between $\gamma=1$ and $\gamma=0.5$, while $\gamma_2=0$ is fixed.
  Then $H(S)=C(S)\le n$ in this case, and $\log N=O(\log n)$.
  Moreover $Q=\frac{Z_{0,1}(G)}{Z_{0,0.5}(G)}\le 2^{\max_{S}C(S)}\le 2^n$.
  Thus $q=\log Q=O(n)$.
  The number of samples required in this step is $O(n\log n\eps^{-2})$.

  Overall, the number of samples required is $O(n(\log n)^2\eps^{-2})$.
  Together with Corollary~\ref{cor:Gibbs-RT}, this implies the claimed run-time.
\end{proof}

\section*{Acknowledgements}

We thank Adrien Kassel for bringing reference \cite{KK17} to our attention.
We also thank Kun He for pointing out an improvement in Lemma~\ref{lem:magic-ratio-edge}.
Part of the work was done while HG was visiting the Institute of Theoretical Computer Science, Shanghai University of Finance and Economics,
and he would like to thank their hospitality.

\bibliographystyle{alpha}
\bibliography{PRS}

\label{lastpage}

\end{document}